\documentclass[11pt,a4paper]{article}

\usepackage[utf8]{inputenc}
\usepackage[T1]{fontenc}
\usepackage{amsmath, amssymb, amsfonts, amsthm}
\usepackage{geometry}
\usepackage{graphicx}
\graphicspath{{figures/}}
\usepackage{booktabs}
\usepackage{hyperref}

\newtheorem{theorem}{Theorem}
\newtheorem{proposition}{Proposition}
\newtheorem{corollary}{Corollary}
\theoremstyle{remark}
\newtheorem{remark}{Remark}

\newcommand{\clo}[1]{{#1}^{(2*)}}   

\title{Symmetry of Spin Systems as Automorphisms of Undirected Weighted Graphs: Realizability Criterion and Complete Taxonomy up to 14 Spins}

\author{Dmitry A. Cheshkov, Dmitry O. Sinitsyn, Artemiy I. Nichugovskiy}

\date{July 24, 2026}

\begin{document}

\maketitle

\begin{abstract}
Exact simulation of high-resolution NMR spectra requires block diagonalization of the spin Hamiltonian, whose dimension grows exponentially with the number of spins $N$; symmetry is the principal tool for taming this growth, yet which permutation groups can occur as the full symmetry group of a scalar-coupled spin system has lacked an exhaustive treatment. Formulating the spin system as an undirected edge-weighted complete graph, we prove an exact realizability criterion: a subgroup of $S_N$ is realizable if and only if it coincides with its symmetrized (undirected) Wielandt 2-closure. In particular, purely rotational symmetry of a single spin ring is impossible, yet chiral spin systems do exist as multi-orbit twisted stacks, and we determine the minimal spin count $\mu^{*}(C_n)$ for every cyclic group, including the counter-intuitive realizations $C_8$ and $C_9$ at $N = 12$. A sequential symmetrization algorithm, completed by an orbit-partition decomposition, yields a provably exhaustive enumeration of all realizable symmetry types up to $N = 14$: the apparently new sequence $a(N) = 1, 1, 3, 8, 11, 27, 36, 90, 131, 282, 394, 948, 1316, 2866$ with the tower law $a(N) = a(N-1) + f(N)$ --- a catalogue of 6112 entries in all, organized by canonical identifiers and a structural grammar extending the Pople nomenclature. Finally, we present a hierarchical methodology for exact block diagonalization without physical approximations: factorization by the conserved total spin projection, Schur--Weyl contraction of magnetically equivalent composites, orbit-weight deduplication of the spin configurations, and isotypic projection over the representations of the factor group, with a uniform treatment of non-abelian groups and complex characters.
\end{abstract}

\section{Introduction}

The exact simulation of high-resolution nuclear magnetic resonance (NMR) spectra remains a cornerstone of modern structural analysis in chemistry and structural biology. The core computational challenge in total lineshape analysis \cite{cheshkov2018} is the exact block diagonalization of the isotropic spin Hamiltonian, which operates in a Hilbert space of dimension $2^N$ for a system of $N$ spin-$\frac{1}{2}$ nuclei. As $N$ increases, the exponential growth of the Hamiltonian matrix makes numerical diagonalization a computationally prohibitive bottleneck.

Exploiting the intrinsic symmetry of the spin system is the most powerful physical approach to mitigate this ``curse of dimensionality.'' Historically, symmetry in NMR has been largely derived from the geometric point-group symmetry of the underlying molecule. However, the isotropic spin Hamiltonian is fundamentally insensitive to the 3D atomic coordinates; it is entirely defined by two sets of parameters: the chemical shifts ($\omega_i$) and the scalar spin--spin coupling constants ($J_{ij}$). Consequently, the proper mathematical object representing an NMR spin system is not a 3D molecular structure, but an \textit{undirected weighted complete graph}, where the vertices are weighted by chemical shifts and the edges are weighted by scalar couplings. The true symmetry of the physical system is precisely the automorphism group of this weighted graph. The computer generation of automorphism groups of edge-weighted graphs --- including NMR coupling graphs, for which the group is a subgroup of the Young subgroup $S_{n_1} \times \cdots \times S_{n_m}$ of the magnetically equivalent vertex classes --- was developed by Balasubramanian \cite{balasubramanian1994}; that line of work computes $\mathrm{Aut}(G)$ for a \emph{given} weighted graph, whereas the structural question we address is the inverse one: which groups can arise at all.

While the classical Pople nomenclature (e.g., $A_2X_2$, $AA^\prime BB^\prime$) \cite{pople1959} has been universally adopted for describing small symmetric systems, it lacks a rigorous algebraic foundation necessary to unambiguously classify complex, highly symmetric topologies. More fundamentally, a critical theoretical question has remained unanswered: out of all possible subgroups of the symmetric permutation group $S_N$, which subgroups can actually be realized as the exact automorphism group of a scalar-coupled spin system? A classical theorem of Frucht \cite{frucht1939} states that every abstract finite group is the automorphism group of \emph{some} graph; the question here is sharper --- which groups are realizable \emph{as permutation groups on the $N$ spins themselves}. Previous empirical observations lacked an exhaustive mathematical treatment, leaving the classification of NMR spin topologies incomplete.

In this paper, we bridge permutation group theory and quantum magnetic resonance to provide a rigorous and comprehensive framework for spin system symmetries. Adapting Wielandt's theory of 2-closures \cite{wielandt1969} to the undirected (scalar-coupled) setting, we formulate an exact realizability criterion: the realizable groups are exactly those closed under the \emph{symmetrized} 2-closure (Section~2). To categorize the resulting catalogue, we propose a formalized grammar extending the Pople notation alongside unambiguous canonical identifiers (Section~3). Based on this Galois correspondence \cite{higman1975}, we develop the enumeration methodology --- from the naive Bell-number baseline through a sequential-symmetrization search to an orbit-partition decomposition theorem with its $N-1$ tower --- that provides a provably exhaustive enumeration of all realizable spin topologies up to $N = 14$ (Section~4), yielding a complete taxonomy of $\sum_{N=3}^{14} a(N) = 6112$ catalogue entries. Section~5 analyses the results: the counting sequence and the machine taxonomy, followed by the structurally interesting cases --- the impossibility of purely rotational symmetry on a single spin ring versus the existence of multi-orbit chiral realizations with their minimal spin counts $\mu^{*}(C_n)$, the revival of the alternating group $A_4$, and the odd-order and simple-group entries built on the Paley tournament and the Fano plane.

The catalogue is complemented by a symmetry-exact three-dimensional visualization: diagonalizing the class-weighted topology matrix of an entry yields embeddings in which every graph automorphism acts as a rigid motion (Section~6). Finally, we translate this structural taxonomy into a highly efficient, general-purpose methodology for exact block diagonalization of the spin Hamiltonian. By hierarchically applying Schur--Weyl duality \cite{weyl1939} to magnetically equivalent composites (twin vertices) \cite{gallai1967, banwell1963} and isotypic projection to the graph's factor group, we present an algorithm capable of handling any realizable topology --- including non-abelian groups and complex characters --- without approximations (Section~7).

\section{Theoretical Framework and the Realizability Criterion}

\subsection{The Spin System as an Edge-Weighted Complete Graph}
Let $\mathcal{S}$ be a quantum-mechanical system of $N$ interacting spin-$\frac{1}{2}$ nuclei. Under the isotropic liquid-state approximation, the system is governed by the scalar-coupled spin Hamiltonian (in units of $\hbar = 1$):
\begin{equation}
    H = \sum_{i \in V} \omega_i I_{z,i} + \sum_{\{i,j\} \in E} J_{ij} (\mathbf{I}_i \cdot \mathbf{I}_j)
\end{equation}
where $V = \{1, 2, \dots, N\}$ represents the set of spin labels, $E = \{\{i,j\} : i,j \in V, i \neq j\}$ is the set of all unordered pairs, $\omega_i$ denotes the chemical shift of the $i$-th spin, and $J_{ij}$ is the isotropic scalar coupling constant between spins $i$ and $j$.

We map this physical system to an undirected, vertex-weighted, edge-weighted complete graph $G = (V, E, \omega, J)$. A permutation of spin labels $\sigma \in S_N$ induces a unitary transformation $U_\sigma$ in the $2^N$-dimensional Hilbert space $\mathcal{H} = \bigotimes_{i=1}^N \mathbb{C}^2$. The operator $U_\sigma$ commutes with the Hamiltonian, $[H, U_\sigma] = 0$, if and only if the permutation preserves the parameter landscape:
\begin{equation}
    \omega_{\sigma(i)} = \omega_i \quad \forall i \in V, \quad \text{and} \quad J_{\sigma(i)\sigma(j)} = J_{ij} \quad \forall \{i,j\} \in E
\end{equation}
Algebraically, this condition is equivalent to stating that $\sigma$ belongs to the automorphism group of the weighted complete graph, $\sigma \in \mathrm{Aut}(G)$. Because the specific numerical values of $\omega_i$ and $J_{ij}$ are prone to experimental variation, the algebraic symmetry of the system is fundamentally a discrete property determined by the partitions of $V$ and $E$ into equivalence classes under the equality of weights.\footnote{For $N \ge 3$ the vertex classes are redundant in the following sense: any partition of $V$ can be induced by edge classes alone, so the enumeration of Section~4 may operate on edge colourings only. The single exception is $N = 2$, where the unique edge cannot distinguish the two vertices; the value $a(2) = 1$ below follows the edge-colouring convention.}

\subsection{Symmetrized 2-Closure and the Realizability Criterion}
A central question in the structural taxonomy of spin systems is identifying which subgroups $H \le S_N$ can occur as the full automorphism group $\mathrm{Aut}(G)$ for some choice of weight functions $\omega$ and $J$. The appropriate closure concept goes back to Wielandt \cite{wielandt1969}, with one essential modification dictated by the physics.

Let $H$ be a permutation group acting on the set $V$. Wielandt's classical construction considers the induced action on \emph{ordered} pairs, $\sigma(i,j) = (\sigma(i), \sigma(j))$; the orbits of this action are the \textit{2-orbits} (orbital relations) of $H$, and the 2-closure $H^{(2)}$ is the largest subgroup of $S_N$ preserving every 2-orbit. Scalar couplings, however, are \emph{symmetric}: $J_{ij} = J_{ji}$ carries no orientation, so the invariant relations available to a spin system live on \emph{unordered} pairs. We therefore work with the induced action of $H$ on $V$ itself (vertex orbits) and on the edge set $E$ of unordered pairs, $\sigma\{i,j\} = \{\sigma(i), \sigma(j)\}$, and define the \textbf{symmetrized 2-closure} ($2^{*}$-closure; the asterisk marks the unordered, symmetrized variant)
\begin{equation}
    \clo{H} \;=\; \bigl\{ \sigma \in S_N :\ \sigma \text{ preserves every vertex orbit and every edge orbit of } H \bigr\}.
\end{equation}
A permutation group is called \emph{$2^{*}$-closed} if $H = \clo{H}$. Clearly $H \le H^{(2)} \le \clo{H}$: symmetrization merges the mutually transposed 2-orbits, so the undirected closure can only be larger.

\begin{theorem}[Realizability criterion]\label{thm:realizability}
A permutation subgroup $H \le S_N$ can be realized as the full symmetry group of a scalar-coupled spin system if and only if $H$ is $2^{*}$-closed, i.e. $H = \clo{H}$.
\end{theorem}

\begin{proof}
($\Rightarrow$) Let $H = \mathrm{Aut}(G)$ for a weighted graph $G = (V, E, \omega, J)$. The level sets of $\omega$ partition $V$ and the level sets of $J$ partition $E$; since every element of $H$ preserves the weights, each level set is a union of $H$-orbits (of vertices and of unordered pairs, respectively). Any $\sigma \in \clo{H}$ preserves every vertex orbit and every edge orbit of $H$, hence preserves every level set of $\omega$ and $J$, hence $\sigma \in \mathrm{Aut}(G) = H$. Combined with the trivial inclusion $H \le \clo{H}$, we obtain $H = \clo{H}$.

($\Leftarrow$) Conversely, let $H$ be $2^{*}$-closed. Construct the \emph{orbit colouring} of $H$: assign one chemical shift value per vertex orbit and one coupling value per edge orbit, all values pairwise distinct. For the resulting graph $G$, a permutation is an automorphism precisely when it preserves every vertex and edge orbit of $H$; by construction this set is exactly $\clo{H} = H$. Thus $\mathrm{Aut}(G) = H$.
\end{proof}

The correspondence $H \mapsto$ (orbit colouring) and (colouring) $\mapsto \mathrm{Aut}$ is a Galois correspondence between the subgroup lattice of $S_N$ and the lattice of edge colourings of $K_N$ \cite{higman1975}; the realizable groups are exactly its closed elements.

\begin{remark}[Directed contrast]\label{rem:directed}
The distinction between the two closures is not pedantic. The cyclic group $C_N$ \emph{is} 2-closed in Wielandt's ordered sense: its 2-orbits distinguish the ``directions'' $+d$ and $-d$ along the ring, and a generic circulant \emph{digraph} has automorphism group exactly $C_N$. It is the symmetrization --- the physical fact that $J_{ij} = J_{ji}$ --- that merges the classes $+d$ and $-d$ and forces the reflection into the closure (Proposition~\ref{prop:ring}). Realizing the ``missing'' groups such as $C_N$ on a single ring would require an antisymmetric pairwise interaction, which the isotropic scalar coupling does not provide.
\end{remark}

\section{Nomenclature and Topology Classification}

\subsection{The Limitation of Classical Notation}
Historically, NMR spin systems have been described using the empirical Pople nomenclature \cite{pople1959}, which assigns letters from the Latin alphabet based on chemical shift differences (e.g., $A, M, X$) and uses primes or subscripts to denote magnetic equivalence and symmetry (e.g., $A_2X_2$, $AA^\prime BB^\prime$). While this notation is highly intuitive for simple molecules, it completely lacks the formal algebraic grammar required to describe the arbitrary finite groups generated by complex weighted graphs. As $N$ increases, relying on ad hoc prime notation becomes fundamentally ambiguous. To address this, we propose a rigorous, machine-readable, three-layer classification scheme that unambiguously identifies and describes any realizable spin topology.

\subsection{Layer 0: The Canonical Identifier}
To ensure stable and unambiguous identification across literature and software implementations, we introduce a canonical identifier for each unique symmetry class, analogous to a CAS Registry Number for chemical compounds. The identifier is defined by the syntax:
\begin{equation}
    N/\mathrm{o}\langle|\mathrm{Aut}|\rangle.\langle k \rangle
\end{equation}
where $N$ is the number of spins, $|\mathrm{Aut}|$ is the order of the full automorphism group of the undirected weighted graph, and $k$ is a deterministic index within the given order, anchored to the published catalogue files. For example, the identifier \texttt{8/o72.3} uniquely denotes an 8-spin system with a group order of 72 (structurally $C_2\langle A_3 A_3^\prime, X X^\prime \rangle$ in the grammar below). Using the sequential symmetrization algorithm described in Section 4 (and, for $N = 14$, the orbit-partition decomposition of Section~\ref{sec:decomposition}), we have computed the complete catalogue for $N$ ranging from 3 to 14 --- $a(N)$ realizable classes at each degree, $\sum_{N=3}^{14} a(N) = 6112$ entries in all --- each assigned a permanent Layer 0 identifier.

\subsection{Layer 1: Magnetic Composites and Factor Graphs}
The physical structure of any symmetric spin graph can be recursively decomposed into fundamental algebraic building blocks. The first step in this decomposition is the identification of ``twin vertices'' --- sets of spins that share identical chemical shifts and identical coupling constants to all other spins in the system. In NMR terminology, these are strictly magnetically equivalent nuclei \cite{corio1966}; in graph theory, twin classes are the modules of the modular decomposition \cite{gallai1967}.

Let the graph be partitioned into magnetic composite classes of sizes $n_g$. The full automorphism group of the system factors as
\begin{equation}\label{eq:factor}
    |\mathrm{Aut}| = \Bigl( \prod_g n_g! \Bigr) \cdot |\Gamma_q|
\end{equation}
where $\Gamma_q$ is the residual factor group acting on the reduced factor graph (each twin class contracted into a single super-vertex). This Layer 1 fingerprint --- the composite sizes $n_g$, the order of the factor group $|\Gamma_q|$, and the vertex orbits --- provides the exact mathematical foundation required for the hierarchical block diagonalization of the Hamiltonian via Schur--Weyl duality. Equation~\eqref{eq:factor} was verified mechanically for all 6112 catalogue entries.

\subsection{Layer 2 and the Extended Structural Grammar}
To bridge the gap between abstract group theory and chemical intuition, we introduce a formal grammar that extends the classical Pople notation to accommodate complex factor groups. This grammar utilizes recursive bracket structures to denote group operations over specific subsystems:
\begin{itemize}
    \item \textbf{Subscripts} denote magnetic composites (e.g., $A_3$ represents a symmetric $S_3$ group of three magnetically equivalent spins).
    \item \textbf{Primes} ($A, A^\prime, A^{\prime\prime}$) denote distinct orbits that share the same chemical shift but have different coupling topologies.
    \item \textbf{Group brackets} such as $C_n\langle \dots \rangle$ or $D_n\langle \dots \rangle$ explicitly indicate the factor group $\Gamma_q$ acting on the enclosed composites; dedicated constructors ($\mathrm{tw}_n$, $\mathrm{st}_n^{0}$, $\mathrm{st}_n^{45}$, $\mathrm{cube}$, \dots) name the recurring twisted-stack and polyhedral motifs.
\end{itemize}
For instance, a system where a dihedral $D_2$ symmetry acts on four identical but distinct composites is unambiguously written as $D_2\langle A_4^\prime \rangle \ast X_4$, directly mapping the algebraic operations to the physical spin topology. This extended grammar is complete: every realizable topology admits a human-readable, yet mathematically strict, structural formula; the full renaming table for the classical 8-spin catalogue (90 entries) is provided in the accompanying repository.

\section{Computational Methodology:\texorpdfstring{\\}{ }Enumerating the Realizable Symmetries}

\subsection{The Naive Baseline: Bell-Number Growth}
Up to renaming of the classes, an edge colouring of $K_N$ is a set partition of its $E = N(N-1)/2$ edges, so the raw search space of the brute-force approach --- enumerate every colouring, compute its automorphism group, deduplicate --- is the Bell number $B_E$ (OEIS A000110). The colourings are generated without repetition as restricted growth strings: edge $k$ may open a new class only one step above the maximum class used by the preceding edges. Mechanically, this is a bank of coupled gears with growing tooth counts, one gear per edge --- an Enigma-like odometer in which the rightmost wheel spins fastest, carries propagate leftward, and each wheel may climb only one step above the maximum reached by the wheels to its left before resetting; one full revolution of the machine visits each of the $B_E$ colourings exactly once. The growth is super-exponential: already for $N = 6$ one must sweep $B_{15} = 1\,382\,958\,545$ colourings to surface $a(6) = 27$ symmetry classes, and $N = 7$ demands $B_{21} \approx 4.7 \cdot 10^{14}$. Brute force is therefore dead on arrival beyond $N = 6$; the constructive search below inverts the strategy --- instead of enumerating all colourings and testing them for symmetry, it grows only the symmetric ones.

\subsection{Constructive Breadth-First Search in the Coloring Space}
To generate an exhaustive catalogue of all physically realizable spin system symmetries (i.e., all $2^{*}$-closed subgroups of $S_N$ up to conjugacy), we developed a sequential symmetrization algorithm. Traditional subgroup enumeration methods are highly redundant for this task, as the vast majority of permutation groups are not $2^{*}$-closed (cf.\ Table~\ref{tab:sequence}). Our approach solves the inverse problem: we enumerate unique orbit colourings (representing the scalar coupling matrices $J$), whose automorphism groups are by construction exactly the realizable groups.

The algorithm employs a constructive breadth-first search (BFS) over the space of canonical graph colourings. The initial state is a completely asymmetric complete graph $K_N$ (a discrete colouring where every edge possesses a unique weight, corresponding to the trivial group). In each round, every colouring of the current frontier is symmetrized by every permutation $\sigma \in S_N$: using a disjoint-set (union--find) structure, the edge classes are merged along the orbits of $\sigma$, producing the coarsest colouring for which $\sigma$ is an automorphism. Each resulting colouring is reduced to a canonical form --- the lexicographically minimal relabelling over all vertex orderings --- so that each symmetry class is stored exactly once. The canonical form, together with the automorphism group $\mathrm{Aut}$ (its order, generators and vertex orbits), is computed by a nauty-style individualization--refinement procedure \cite{mckay2014}: colour refinement to an equitable vertex partition, then a refinement-guided ordering search in which the automorphisms discovered so far prune the search tree by orbit, so that highly symmetric colourings --- the costliest case for a naive ordering scan --- are the cheapest here; $|\mathrm{Aut}|$ follows from the generators by Schreier--Sims without enumerating group elements. Two optimizations are exact: in the first round it suffices to take one representative $\sigma$ per cycle type (conjugate generators yield conjugate colourings), and symmetrizing by $\sigma$ is equivalent to symmetrizing by the whole cyclic group $\langle \sigma \rangle$ at once. In lattice-theoretic terms, the algorithm computes joins: round $r$ produces the closures of all subgroups generated by $r$ ``cyclic-closure'' steps, and the set of realizable groups is the join-closure of the cyclic ones.

\subsection{Mathematical Proof of Completeness}
A fundamental advantage of this constructive approach is its built-in stopping criterion, which guarantees the completeness of the generated catalogue. By the McIver--Neumann bound \cite{mciver1987} (cf.\ also Jerrum's $N-1$ bound \cite{jerrum1986}), any subgroup of $S_N$ can be generated by at most $\max(2, \lfloor N/2 \rfloor)$ elements. Consequently, every realizable group is reached after at most $\lfloor N/2 \rfloor$ productive symmetrization rounds, and the first empty round certifies exhaustion.

The bound is tight in practice. For $N = 8$ the catalogue is completed in exactly 4 productive rounds; the unique round-4 entry is the elementary abelian group $(C_2)^4$ in its \emph{regular-type} action on two tetrads (two commuting Klein four-groups), which genuinely requires four generators. Notably, the ``Young-type'' $(C_2)^4$ generated by four disjoint transpositions is reached already in round 3: its rank-3 even-weight subgroup has the same edge orbits, so the fourth generator is absorbed by the closure. For $N = 12$ the enumeration terminates after 5 productive rounds, within the guaranteed bound of 6.

\subsection{High-Performance Implementation}
Because the space of raw colourings grows rapidly, reaching $N = 12$ required an optimized implementation in C++. With the individualization--refinement canonicalizer above, the dominant cost is no longer canonization but the coarsening and deduplication of the raw colourings produced by the sweep: within each round the frontier work is parallelized with OpenMP (dynamic scheduling), raw colourings are hashed and deduplicated before canonization, and duplicate canonical forms are discarded through global hash tables. The complete enumeration up to $N = 12$ was executed on a dual-socket Intel Xeon Gold 6542Y node, with the $N = 12$ level completing in under two hours. For $N = 13$ a distributed MPI+OpenMP version of the enumerator was developed, which shards the deduplication tables across ranks by colouring hash, exchanges raw colourings in size-capped batches, and checkpoints after every round. The $N = 13$ level (185 million raw colourings exchanged in round 2 alone) completed in 23.4 hours on two such nodes, giving $a(13) = 1316$ in five productive rounds, within the guaranteed bound of six.

\subsection{Orbit-Partition Decomposition and Cross-Validation}\label{sec:decomposition}
The breadth-first search is bounded by the size of the raw-colouring space it must sift, and the direct alternative --- enumerate every subgroup class of $S_N$ and retain the $2^{*}$-closed ones --- is bounded by the cost of the subgroup lattice itself: for $S_{14}$ the cyclic-extension lattice \cite{gap,hulpke2005} exhausts memory before enumerating its $75154$ classes. To reach $N = 14$ we assemble the catalogue orbit-partition by orbit-partition, a route that materializes neither the full colouring space nor the full $S_N$ lattice. The route rests on the following construction.

\begin{theorem}[Exhaustive construction of the realizable types]\label{thm:construction}
Every realizable symmetry type of an $N$-spin system is a $2^{*}$-closed subgroup of $S_N$ taken up to spin renumbering, and the following construction yields the exhaustive list: (i) run over the integer partitions $N = \lambda_1 + \dots + \lambda_k$; (ii) for each partition, over the tuples $(T_1, \dots, T_k)$ of transitive permutation groups of degrees $\lambda_1, \dots, \lambda_k$; (iii) for each tuple, over its subdirect products (Goursat) --- the subgroups $G \le T_1 \times \dots \times T_k$ projecting onto every factor surjectively; (iv) retain the $2^{*}$-closed $G$; (v) identify the results up to spin renumbering.
\end{theorem}

\begin{proof}
Completeness: the orbits of a $2^{*}$-closed group $G$ form a partition $\lambda \vdash N$; $G$ is transitive on each of its orbits, hence a subdirect product of its projections, which are transitive groups of degrees $\lambda_i$; therefore $G$ arises at step (iii) and survives steps (iv)--(v). Soundness: whatever passes step (iv) is $2^{*}$-closed by construction.
\end{proof}

The second ingredient reduces the work to the fixed-point-free part of each level:

\begin{theorem}[Tower decomposition]\label{thm:tower}
For $N \ge 4$, $a(N) = a(N-1) + f(N)$, where $f(N)$ is the number of realizable classes at level $N$ whose groups have no fixed vertex. Equivalently: every realizable group with a fixed vertex is a realizable group of $N-1$ spins extended by an isolated spin, and this embedding is injective on classes.
\end{theorem}

By the tower decomposition (Theorem~\ref{thm:tower}) it suffices to construct $f(N)$, the fixed-point-free classes; the remaining $a(N-1)$ classes are the level-$(N-1)$ catalogue with one isolated vertex appended. A fixed-point-free realizable group has an orbit partition $\lambda \vdash N$ with every part $\ge 2$; it lies in the Young subgroup $S_{\lambda_1} \times \dots \times S_{\lambda_k}$ and is transitive on each block. For each such $\lambda$ we enumerate exactly these block-transitive subgroups: for a single block $[N]$ they are the transitive groups of degree $N$ (a tabulated library); for a balanced two-block partition $[a,b]$ they are the subdirect products (Goursat) of a transitive group on each block, read directly from the small transitive-group libraries rather than from the lattice of $S_a \times S_b$ --- decisive for the balanced partitions, which dominate the full-lattice cost; for three or more blocks the Young subgroup is a product of small symmetric groups whose lattice is itself tractable. Every candidate is put through the closure test of Section~4.2 (form its edge orbit-colouring, compute $\mathrm{Aut}$ with the same individualization--refinement canonicalizer, keep it iff $|\mathrm{Aut}| = |G|$) and deduplicated by canonical form, which also merges the conjugates that arise from permuting equal blocks. This gives $f(14) = 1550$ and hence $a(14) = a(13) + f(14) = 2866$, of which the single-orbit part is $10$ of the $63$ transitive groups of degree $14$. The decomposition is exhaustive by Theorem~\ref{thm:construction}, specialized to the unit-free partitions: every subgroup has a well-defined orbit partition, and the block-transitive subgroups of the associated Young subgroup are precisely the groups realizing that partition.

One point about Theorem~\ref{thm:construction} deserves emphasis: the transitive components at step (ii) must range over \emph{all} transitive groups, not merely the $2^{*}$-closed ones. The projection of a $2^{*}$-closed group onto one of its orbits is transitive but in general not closed --- for the chiral entry \texttt{6/o3.1} both projections are $C_3$ on three points, whose closure is $S_3$ --- so restricting the components to closed groups would lose the entire chiral family; the closure test is applied once, to the assembled subdirect product (necessarily so, since a subdirect product of closed components may itself fail to be closed, e.g.\ $S_3 \times_{C_2} S_3$ of order 18 at $N = 6$). Equivalently, the search can be organized in two tiers: the induced colouring of an orbit determines only the \emph{closure} of the projection ($T$ and its closure share the same orbit colouring), so one first runs over all $2^{*}$-closed transitive groups per orbit and then over the transitive subgroups with that closure as the actual components, pinned down by the cross-orbit classes. In the single-orbit case the second tier collapses: with no cross-orbit classes a non-closed component cannot be pinned, the final test forces the component to equal its closure, and the enumeration reduces to the $2^{*}$-closed transitive groups themselves --- the $10$ of the $63$ transitive groups of degree $14$.

Integer partitions of $N$ are, in fact, the common skeleton of the constructive search and the decomposition. In the breadth-first search they index the conjugacy classes (cycle types) of the round-one generators; in the decomposition they index the vertex-orbit structures of the target groups. The link is exact: the $2^{*}$-closure preserves vertex orbits, and the orbits of a cyclic group are the cycles of its generator, so the orbit partition of $\mathrm{cl}\langle\sigma\rangle$ equals the cycle type of $\sigma$. Consequently the catalogue after round one is in bijection with the partitions of $N$ --- the round-one totals are the partition numbers $p(N)$: $5, 7, 11, 15, 22, 30, 42, 56, 77, 101$ entries for $N = 4$--$13$ --- and each $\mathrm{cl}\langle\sigma_\lambda\rangle$ is precisely the cyclic block-transitive seed of the stratum $\lambda$ that the decomposition completes inside its Young subgroup. The tower decomposition itself mirrors the elementary partition identity $p(N) = p(N-1) + p_{\ge 2}(N)$, with unit parts corresponding to fixed points and unit-free partitions to the fixed-point-free classes: at $N = 14$, $135 = 101 + 34$, and exactly $34$ strata enter the computation of $f(14)$.

The decomposition also furnishes strong independent validation. Applied to $N \le 13$ it reproduces $f(4), \dots, f(13)$ term for term; and for every two-block partition it yields byte-identical closed-colouring sets whether the block-transitive subgroups are taken from the full Young lattice or from subdirect products. A third, wholly separate route corroborates the lower levels: enumerating all conjugacy classes of subgroups of $S_N$ in GAP \cite{gap,hulpke2005} and filtering by $2^{*}$-closure reproduces $a(3), \dots, a(13)$ exactly --- two algorithms of opposite character, symmetrize-then-search and enumerate-then-filter, agreeing on the entire sequence. The force of this agreement lies in the disjoint dependencies of the routes: the breadth-first search is entirely self-contained --- its input is the single integer $N$, it manipulates nothing but edge colourings and permutations, and its termination certificate is internal --- whereas the other two routes lean on independent external machinery, the subgroup-lattice algorithms of GAP and the tabulated transitive-group library, respectively. The three methods therefore share no failure modes, and their term-for-term agreement cross-certifies each against the dependencies of the others.

\section{Results and Discussion}

\subsection{The Counting Sequence and the Tower Decomposition}
Table~\ref{tab:sequence} presents the complete counting results: $a(N)$, the number of realizable symmetry classes ($2^{*}$-closed subgroups of $S_N$ up to conjugacy), against $s(N)$, the total number of subgroup classes of $S_N$ (OEIS A000638 \cite{oeisA000638}). The realizable groups are a rapidly thinning minority: at $N = 12$ only 948 of 10723 subgroup classes survive the closure condition, at $N = 13$ only 1316 of 20832, and at $N = 14$ only 2866 of 75154. To the best of our knowledge the sequence $a(N)$ is new; it will be submitted to the OEIS.

\begin{table}[t]
\centering
\caption{Realizable symmetry classes $a(N)$ versus all subgroup classes $s(N) = $ A000638$(N)$, and the fixed-point-free increment $f(N)$ of the tower decomposition (Theorem~\ref{thm:tower}). The value $a(2) = 1$ follows the edge-colouring convention of Section 2.1.}
\label{tab:sequence}
\begin{tabular}{lrrrrrrrrrrrrrr}
\toprule
$N$    & 1 & 2 & 3 & 4  & 5  & 6  & 7  & 8   & 9   & 10   & 11   & 12 & 13 & 14 \\
\midrule
$a(N)$ & 1 & 1 & 3 & 8  & 11 & 27 & 36 & 90  & 131 & 282  & 394  & 948 & 1316 & 2866 \\
$s(N)$ & 1 & 2 & 4 & 11 & 19 & 56 & 96 & 296 & 554 & 1593 & 3094 & 10723 & 20832 & 75154 \\
$f(N)$ & --- & --- & --- & 5 & 3 & 16 & 9 & 54 & 41 & 151 & 112 & 554 & 368 & 1550 \\
\bottomrule
\end{tabular}
\end{table}

The values $f(4), \dots, f(14) = 5, 3, 16, 9, 54, 41, 151, 112, 554, 368, 1550$ were verified directly: the machine taxonomy marks each entry as fixed-point-free or not, and the counts reproduce the differences $a(N) - a(N-1)$ exactly. The tower is in fact the backbone of the $N = 14$ computation, which is assembled from $f(14)$ by the orbit-partition decomposition of Section~\ref{sec:decomposition}. The alternation visible in $f$ (each odd level is poorer than its even predecessor) reflects the parity constraints of Section~\ref{sec:chiral}: many symmetry motifs, including all new cyclic groups (Corollary~\ref{cor:parity}), can debut only at even $N$.

\subsection{Structure of the Taxonomy}\label{sec:taxonomy}
For every one of the 6112 catalogue entries, the accompanying machine taxonomy records the composite sizes, $|\Gamma_q|$, the vertex orbits, fixed-point-freeness, cyclicity, the maximal element order, abelianity and ambivalence. Table~\ref{tab:stats} summarizes several structural counts.

\begin{table}[t]
\centering
\caption{Structural statistics of the catalogue. ``fpf'' = fixed-point-free classes ($= f(N)$ for $N \ge 4$); ``cyclic''/``abelian'' = classes with cyclic/abelian $\mathrm{Aut}$; ``trivial twins'' = classes with no magnetic equivalence ($\prod n_g! = 1$); ``trivial $\Gamma_q$'' = classes whose symmetry is exhausted by magnetic equivalence ($|\Gamma_q| = 1$).}
\label{tab:stats}
\begin{tabular}{rrrrrrr}
\toprule
$N$ & $a(N)$ & fpf & cyclic & abelian & trivial twins & trivial $\Gamma_q$ \\
\midrule
3  & 3   & 1   & 2  & 2   & 1   & 3  \\
4  & 8   & 5   & 3  & 5   & 3   & 5  \\
5  & 11  & 3   & 3  & 5   & 4   & 7  \\
6  & 27  & 16  & 5  & 11  & 10  & 11 \\
7  & 36  & 9   & 5  & 11  & 11  & 15 \\
8  & 90  & 54  & 8  & 27  & 31  & 22 \\
9  & 131 & 41  & 10 & 29  & 42  & 30 \\
10 & 282 & 151 & 14 & 61  & 87  & 42 \\
11 & 394 & 112 & 16 & 64  & 99  & 56 \\
12 & 948 & 554 & 26 & 152 & 279 & 77 \\
13 & 1316 & 368 & 28 & 161 & 334 & 101 \\
14 & 2866 & 1550 & 40 & 348 & 736 & 135 \\
\bottomrule
\end{tabular}
\end{table}

Several trends are visible. The fraction of classes with non-trivial magnetic equivalence stays dominant but declines slowly ($\sim 74\%$ at $N = 14$), while the classes whose symmetry is \emph{entirely} magnetic (trivial $\Gamma_q$ --- the classical ``first-order'' systems such as $A_2 M_3 X$) become a small minority ($135$ of $2866$ at $N = 14$): topological symmetry, invisible to first-order multiplet analysis, is the rule rather than the exception among symmetric topologies. Cyclic (chiral) groups remain rare --- 40 classes at $N = 14$ --- consistently with the strong constraints of Section~\ref{sec:chiral}.

\subsection{Algebraic Landscape: Products, Extensions, and Wreath Motifs}
The allowed $2^{*}$-closed subgroups exhibit highly constrained algebraic structures dictated by the modularity of the spin graphs. The simplest are \textbf{direct products} (e.g., $G = H_1 \times H_2$), corresponding to disjoint or uniformly coupled subsystems (e.g., $A_m X_n$). At the next level, Eq.~\eqref{eq:factor} equips every entry with the exact sequence
\begin{equation}
    1 \longrightarrow \prod_g S_{n_g} \longrightarrow \mathrm{Aut}(G) \longrightarrow \Gamma_q \longrightarrow 1,
\end{equation}
the kernel being the magnetic (Young-type) part and the quotient the topological part. When $\Gamma_q$ permutes $k$ equal composites of size $m$, the extension specializes to the \textbf{wreath product} $S_m \wr K$: catalogue examples include $(A_2)_3 = S_2 \wr S_3$ of order 48, $(A_3)_2 = S_3 \wr S_2$ of order 72, $(A_4)_2 = S_4 \wr S_2$ of order 1152, and the polyhedral $\mathrm{cube}(A^\prime 8)$ family. In general, however, the extension need not be of wreath type, and the factor group itself ranges over dihedral, polyhedral and non-split motifs; it is precisely this variety that necessitates the general isotypic machinery of Section~7 rather than any ad hoc treatment of special cases.

This highlights the physical significance of the \textbf{factor graph}. Once all magnetically equivalent twin vertices are contracted via Schur--Weyl duality, the residual factor graph represents the pure topological skeleton of the spin system, and its automorphism group $\Gamma_q$ governs the macroscopic complexity of the energy level structure. While purely chiral $C_n$ groups are forbidden for single rings (Proposition~\ref{prop:ring}), they do appear --- at the minimal spin counts $\mu^{*}(C_n)$ of Section~\ref{sec:chiral} --- and complex factor groups with dihedral, polyhedral and non-abelian structures emerge systematically for $N \ge 6$.

\subsection{Non-Realizability of Purely Chiral Single-Ring Symmetries}

\begin{proposition}[Single ring closes to the dihedral group]\label{prop:ring}
Let $N \ge 3$ and let the coupling matrix $J$ be invariant under the cyclic rotation $c = (1\,2\,\dots\,N)$ acting on a single $N$-vertex ring with equal shifts. Then $J$ is invariant under the full dihedral group $D_N$; consequently $\clo{C_N} \supseteq D_N$ and the purely rotational group $C_N$ is not realizable on a single ring.
\end{proposition}

\begin{proof}
The edge orbits of $C_N$ on unordered pairs are the \emph{distance classes} $\{\{i, i+d\} : i \in \mathbb{Z}_N\}$, $d = 1, \dots, \lfloor N/2 \rfloor$, determined by the cyclic distance $d(i,j) = \min(|i-j|, N - |i-j|)$. The reflection $\tau(i) = -i \pmod N$ preserves cyclic distances, hence preserves every distance class, hence $\tau \in \clo{C_N}$.
\end{proof}

Therefore, any scalar-coupled spin matrix $J$ that is invariant under the cyclic group $C_N$ acting on a single ring is mathematically forced to be invariant under the full dihedral group $D_N$: on a single ring, rotational symmetry is always accompanied by specular reflection symmetry, due to the scalar (non-directional) nature of the $J$-couplings.

\subsection{Minimal Realizations of Chiral (Cyclic) Symmetry}\label{sec:chiral}

Proposition~\ref{prop:ring} does \emph{not} forbid chiral spin systems: it only forbids single-ring realizations. The canonical chiral realization of $C_n$ is a pair of \emph{twisted $n$-gons}: two $n$-rings rotated synchronously by the generator, with the $n$ inter-ring ``skew matching'' classes labelled by the difference $d \in \mathbb{Z}_n$. The intra-ring distance classes are automatically reflection-symmetric, but a reflection must act on the inter-ring difference classes as $d \mapsto c - d$ for some constant $c$, which cannot fix all $n \ge 3$ classes; generic (pairwise distinct) inter-ring labels therefore kill every reflection, and the closure is exactly $C_n$. Geometrically this is an antiprism at a generic twist angle --- a \emph{chiral configuration} of spins (Figure~\ref{fig:embeddings}a,c).

The twisted pair, however, is not always minimal. The general construction is a \textbf{family of quotient rings}: the generator rotates by $+1$ several rings of sizes $m_t \mid n$ with $\operatorname{lcm}\{m_t\} = n$ (faithfulness). Between two rings of sizes $m_s$ and $m_t$ there are exactly $\gcd(m_s, m_t)$ skew classes, indexed by a difference in $\mathbb{Z}_{\gcd(m_s,m_t)}$; any reflection acts on them as $k \mapsto -k + c$, which cannot preserve all classes once $\gcd(m_s, m_t) \ge 3$. Conversely, a ring of size $\ge 3$ that belongs to \emph{no} pair with $\gcd \ge 3$ can be flipped independently of the rest (for example, the family $\{4,3,3\}$ fails to realize $C_{12}$: the square flips on its own). This yields a complete criterion:

\begin{proposition}[Quotient-ring criterion]\label{prop:family}
A family of quotient rings of sizes $\{m_t\}$, $m_t \mid n$, $\operatorname{lcm}\{m_t\} = n$, with generic labels realizes exactly $C_n$ if and only if
(i) every ring of size $\ge 3$ belongs to at least one pair with $\gcd(m_s, m_t) \ge 3$, and
(ii) the group of synchronized rotations $\{(r_t) : r_s \equiv r_t \bmod \gcd(m_s, m_t)\ \forall s,t\}$ has order exactly $n$.
\end{proposition}

For a prime power $n = p^e$ the criterion gives the minimal spin count
\begin{equation}\label{eq:mu}
    \mu^{*}(C_{p^e}) \;=\; p^e + \min\{\, p^f : 3 \le p^f \le p^e \,\},
    \qquad \mu^{*}(C_2) = 2,
\end{equation}
i.e. $p^e + p$ for odd $p$, $2^e + 4$ for $p = 2$, $e \ge 3$, and the twisted pair $4+4$ for $C_4$. For composite $n$ the minimum is an additive assembly of the prime-power components on disjoint supports, $\mu^{*}(C_n) = \sum_p \mu^{*}(C_{p^{e_p}})$.

Two consequences of Eq.~\eqref{eq:mu} are counter-intuitive and were, in fact, first observed in the machine-generated catalogue of the present work: $\mu^{*}(C_8) = 8 + 4 = 12$ and $\mu^{*}(C_9) = 9 + 3 = 12$, \emph{not} $16$ and $18$ as the naive two-ring formula $2 p^e$ would suggest. The $N = 12$ catalogue contains exactly one cyclic group of order 8 (entry \texttt{12/o8.44}, vertex orbits $8+4$) and exactly one of order 9 (entry \texttt{12/o9.3}, orbits $9+3$; Figure~\ref{fig:embeddings}d); both are absent for all $N \le 11$, and both were verified by an independent backtracking enumeration of the automorphisms and their element orders ($C_8$: orders $1,2,4,4,8,8,8,8$; $C_9$: orders $1,3,3,9,9,9,9,9,9$). The catalogue likewise confirms $\mu^{*}(C_6) = 8$ (orbits $2+3+3$), the quotient-ring realization $6+3$ of $C_6$ at $N = 9$, and $\mu^{*}(C_{10}) = 12$ (orbits $5+5+2$, entry \texttt{12/o10.5}). The $N = 14$ level sharpens the picture further: the cyclic groups debuting there are exactly $C_7$, $C_{12}$ and $C_{18}$, all with $\mu^{*} = 14$ --- $7+7$ for the prime $C_7$, and the additive assemblies $8+6$ ($C_4$ with $C_3$) and $2+12$ ($C_2$ with $C_9$) for the composites --- so that a cyclic group of order 7, absent for all $N \le 13$, first appears at exactly $N = 2\cdot 7$.

\begin{corollary}[Parity of chiral debuts]\label{cor:parity}
$\mu^{*}(C_n)$ is even for every $n \ge 2$. Hence new cyclic symmetry groups debut only at even spin counts: the observed first appearances of $C_2$, $C_3$, $C_4$, $C_6$, $C_5$, $C_8$, $C_9$, $C_{10}$, $C_7$, $C_{12}$, $C_{18}$ are $N = 2$, $6$, $8$, $8$, $10$, $12$, $12$, $12$, $14$, $14$, $14$, respectively.
\end{corollary}

The corollary made a falsifiable prediction, tested against the subsequently
completed $N = 13$ and $N = 14$ levels. At the odd level $N = 13$ no new cyclic
group may appear, and none does: the cyclic groups present are exactly those
inherited from $N \le 12$ ($C_1, \dots, C_6, C_8, C_9, C_{10}$), while the level
merely debuts the dihedral $D_{13}$ (order 26, a single 13-ring) and still
contains no groups of order 7, 11 or 13. At the even level $N = 14$ the rule
instead permits debuts, and precisely three appear --- $C_7$, $C_{12}$, $C_{18}$,
all with $\mu^{*} = 14$ (Section~\ref{sec:chiral}) --- so that a cyclic group
of order 7 enters the catalogue for the first time, at exactly $N = 2\cdot 7$.

\begin{figure}[!tp]
\centering
\includegraphics[width=\textwidth]{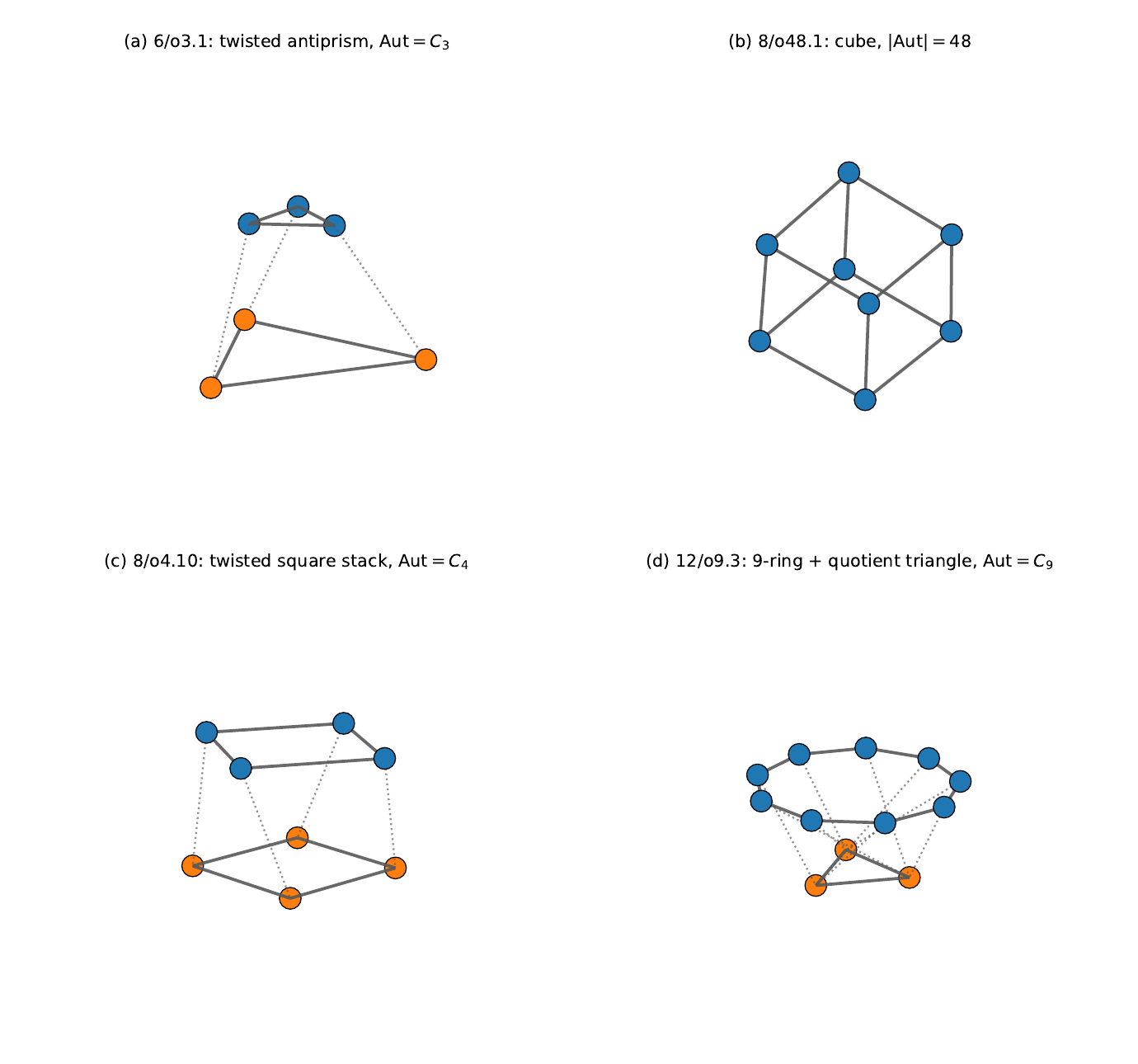}
\caption{Chiral and achiral realizable topologies. Vertex colour encodes the
vertex orbit; solid lines show the ring/cage coupling classes, dotted lines
one representative inter-orbit (``skew matching'') class; all remaining
couplings of the complete graph carry their own class labels and are omitted
for clarity. (a),(b) are symmetry-exact spectral embeddings produced by the
accompanying \texttt{topology2xyz} tool (invariant eigenspace triples of the
class-weighted topology matrix, Section~\ref{sec:embedding}); (c),(d) are the canonical twisted-stack
constructions drawn schematically. (a) The minimal chiral system: two twisted
triangles, $\mathrm{Aut} = C_3$ at $N = 6$. (b) The cube at $N = 8$,
$|\mathrm{Aut}| = 48$. (c) Twisted square stack, $\mathrm{Aut} = C_4$ at
$N = 8$. (d) The realization $C_9$ at $N = 12$ (orbits $9+3$), which refutes
the naive $2p^e$ bound; note that this action is not a rigid three-dimensional
rotation group --- the generator advances the 9-ring by one step but the
triangle by $120^\circ$ --- so the depicted twist is schematic.}
\label{fig:embeddings}
\end{figure}

\subsection{Alternating Components and the Revival of \texorpdfstring{$A_4$}{A4}}\label{sec:alternating}
The two-tier reading of Theorem~\ref{thm:construction} (Section~\ref{sec:decomposition}) sorts the transitive components by their closures and asks which \emph{non-closed} components are actually pinned down in the assembled catalogue. For the cyclic components the answer is the chiral family of Section~\ref{sec:chiral}. The natural alternating groups are the next candidates: $A_n$ is transitive on unordered pairs, its orbit colouring of $K_n$ is monochrome, so $\clo{A_n} = S_n$ --- can a realizable group nevertheless project onto one of its orbits as $A_n$?

A direct scan of all 6112 entries settles the question. For every vertex orbit of size $n \ge 3$ free of magnetic twins, the restriction of $\mathrm{Aut}$ to the orbit was enumerated explicitly; a subgroup of order $n!/2$ is necessarily $A_n$, the unique index-2 subgroup of $S_n$, so no odd permutation can occur in it --- which also disqualifies twin-carrying orbits outright, since a twin swap restricts to a transposition. Projections equal to $A_3 = C_3$ are plentiful: 561 orbit instances, the chiral family again, beginning with the twisted triangles of entry \texttt{6/o3.1}. Projections equal to $A_n$ with $n \ge 5$ never occur. And $A_4$ revives exactly twice: entries \texttt{13/o12.10} (vertex orbits $4+6+3$) and \texttt{14/o12.32} (the same with an isolated spin appended), both with $\mathrm{Aut} \cong A_4$ of order 12 acting faithfully.

The mechanism is the chiral twist of Section~\ref{sec:chiral} lifted to a quotient group. All three orbits of \texttt{13/o12.10} are built from a single 4-element set: its four points, its six unordered pairs, and its three partitions into two disjoint pairs. The point--pair coupling is the incidence relation (two classes), which locks these two orbits together; the point--partition coupling is forced monochrome ($A_4$ is transitive on point--partition incidences); and the entire parity selection resides in the pair--partition coupling, which carries \emph{three} classes --- ``pair inside its partition'' plus two \emph{distinguished} transversal classes. Under $S_4$ the transversal incidences form a single class: the transposition $(ab)$ fixes the partition $ab|cd$ but exchanges $ac|bd$ with $ad|bc$, and with them the two classes. Distinguishing them therefore eliminates every odd permutation: $S_4$ acts on the partition triple through $S_4/V_4 \cong S_3$, the colouring pins that image down to $C_3$, and the preimage of $C_3$ in $S_4$ is exactly $A_4$. The trick is available to $A_4$ alone among the alternating groups: for $n \ge 5$ the group $A_n$ is simple, so there is no small quotient to pin, and the $A_n$-orbits on couplings between point- and pair-type orbits coincide with the $S_n$-orbits --- a transposition of two points outside any given incidence configuration exists whenever $n \ge 5$, while $n = 4$ is precisely the tight case in which the configuration exhausts the points and the class splits (larger $A_n$-actions, on $k$-subsets or coset spaces, already exceed the $N \le 14$ budget). Each orbit of \texttt{13/o12.10} taken separately is non-closed --- $\clo{A_4}$ on the points is $S_4$, $\clo{C_3} = S_3$ on the partitions, and on the pairs $\clo{A_4} = C_2 \wr S_3$ of order 48, the octahedral colouring of $K_6$ --- yet all three components come alive in the ensemble: the sharpest illustration in the catalogue of the remark following Theorem~\ref{thm:construction}. Incidentally, the \emph{regular} action of $A_4$ is closed already at $N = 12$ (the single-orbit entry \texttt{12/o12.7}; its catalogue neighbour \texttt{12/o12.9} is the regular $D_6$), but a natural 4-point alternating orbit debuts only at $N = 13$. Since $A_4$ is non-ambivalent --- its 3-cycles split into two mutually inverse conjugacy classes --- both entries carry genuinely complex character tables and exercise the complex-character branch of the L2b machinery of Section~7.

\subsection{Odd Order, the Paley Tournament, and the Fano Plane}\label{sec:zoo}
The question of the previous subsection generalizes: decompose every entry into its Goursat components --- for each vertex orbit of $\mathrm{Aut}$, the projection $T_i$ (the transitive component actually pinned by the gluing) and its closure $\clo{T_i}$ (the deck group of the two-tier scheme of Section~\ref{sec:decomposition}) --- and ask which non-closed components are ever realized. A sweep of all 6112 entries gives a closed answer: phantom projections occur in roughly a thousand orbit instances, but in only 18 species $(n, |T|, |\clo{T}|)$. Besides the chiral cyclic family $C_3$--$C_9$ (561 instances of $C_3$ alone) and $A_4$, the list contains their ``thickened'' versions carried by magnetic multiplets --- entry \texttt{12/o648.1} realizes the wreath phantom $S_3 \wr C_3$ inside $S_3 \wr S_3$, a twisted triangle whose vertices are magnetic triples --- the abelian $C_4 \times C_2$ and $C_3 \times C_3$ (entry \texttt{12/o9.2}: the regular $C_3 \times C_3$, a \emph{chiral $3 \times 3$ torus} whose closure, the generalized dihedral group of order 18, adds the lattice inversion $x \mapsto -x$), the group $C_3 \times S_3$, the rotational $S_4$ (entry \texttt{10/o24.8}: a 6-vertex octahedral orbit cut from its full octahedral closure of order 48 down to the 24 rotations by a tetrahedral companion orbit --- the spin-system analogue of $O_h \to O$), and two degree-7 jewels at $N = 14$.

Entry \texttt{14/o21.1} realizes $F_{21} = C_7 \rtimes C_3$ of order 21 --- the first non-abelian group of odd order in the catalogue. Both 7-orbits are internally monochrome, and necessarily so: $F_{21} \le S_7$ has $\clo{F_{21}} = S_7$, since its multiplier subgroup $\langle 2 \rangle \le \mathbb{Z}_7^{\times}$ fuses all three distance classes $\{\pm d\}$. The entire structure resides in the 49 cross-couplings: a perfect matching (7 edges) plus two 21-edge classes formed by the quadratic and non-quadratic residue differences --- the symmetrized \emph{Paley tournament} on 7 vertices (Figure~\ref{fig:paleyfano}a). A tournament orientation cannot be written inside one orbit, where unordered pairs identify $d$ with $-d$; between two orbits it can, because $\{i, j'\}$ and $\{j, i'\}$ are distinct edges of $K_{14}$. Containing no involution, $F_{21}$ is non-ambivalent, and its character table is intrinsically complex (Section~7).

\begin{figure}[!tp]
\centering
\includegraphics[width=\textwidth]{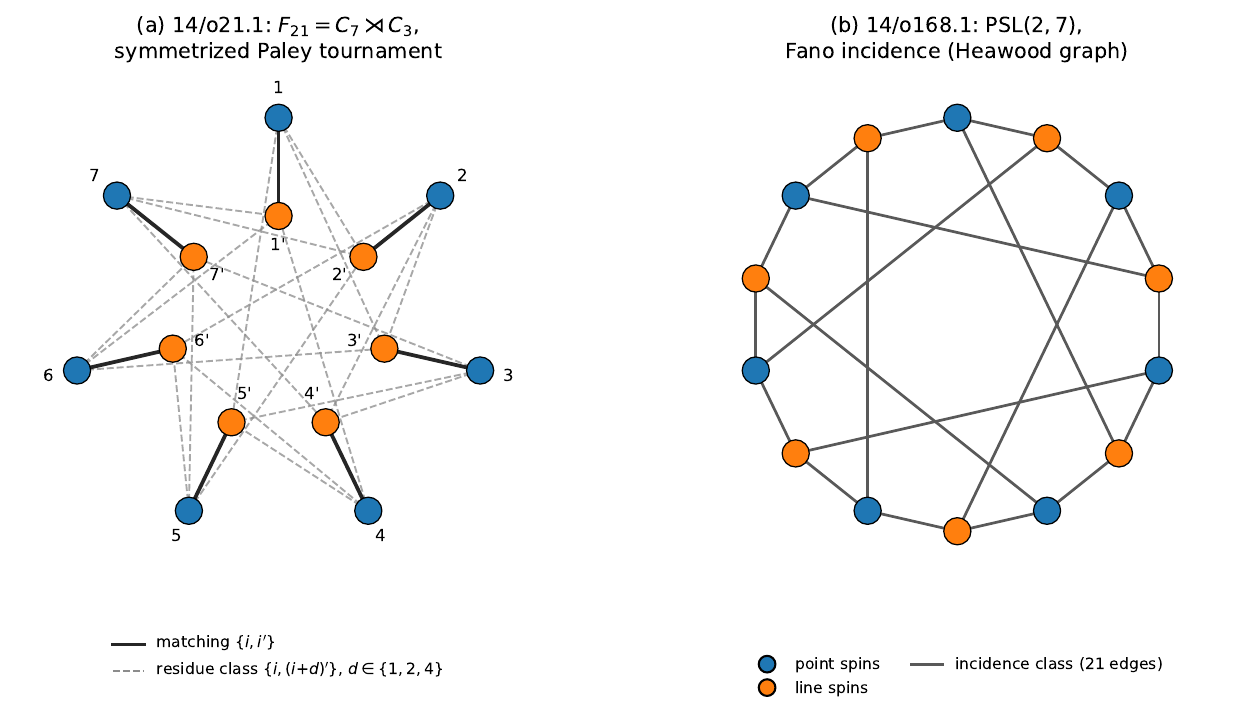}
\caption{The two degree-7 jewels of the catalogue, drawn in their canonical
difference models. (a) Entry \texttt{14/o21.1}, $\mathrm{Aut} = F_{21} =
C_7 \rtimes C_3$ of order 21: two internally monochrome 7-orbits (blue and
orange); the 49 cross-couplings split into the perfect matching
$\{i, i^{\prime}\}$ (solid), the residue class $\{i, (i{+}d)^{\prime}\}$ with
$d \in \{1,2,4\}$ (dashed), and the complementary class $d \in \{3,5,6\}$
(omitted for clarity) --- the symmetrized Paley tournament. Every reflection
and every deck swap maps the residue class onto its complement, so
distinguishing the two classes leaves exactly the chiral Frobenius group.
(b) Entry \texttt{14/o168.1}, $\mathrm{Aut} = \mathrm{PSL}(2,7)$ of order
168: point spins and line spins of the Fano plane; the single 21-edge
incidence class forms the Heawood graph, drawn in its standard LCF
$[5,-5]^7$ form (all remaining couplings are constant within their classes
and omitted).}
\label{fig:paleyfano}
\end{figure}

Entry \texttt{14/o168.1} realizes the simple group $\mathrm{PSL}(2,7) \cong \mathrm{GL}(3,2)$ of order 168: seven ``point'' spins and seven ``line'' spins, both orbits monochrome, the cross-block two-coloured by the incidence relation of the Fano plane $\mathrm{PG}(2,2)$ --- the bipartite Heawood graph (Figure~\ref{fig:paleyfano}b). Both projections are the natural degree-7 action, a phantom of index 30 in its closure $S_7$. The mechanism sharpens the diagonal gluing: the subdirect product is twisted by the \emph{outer} automorphism (point--line duality), and the cross-orbit classes are non-trivial precisely because a point stabilizer ($S_4$, of index 7) is intransitive on the lines (orbits $3 + 4$). For the alternating groups this loophole is unavailable --- their outer twist is realized inside $S_n$ --- which is why the alternating series stops at $A_4$ (Section~\ref{sec:alternating}). The sister entry \texttt{14/o336.1} fuses the two sides into a single 14-vertex orbit: the duality is restored and the group doubles to $\mathrm{PGL}(2,7)$, one of the ten closed transitive groups of degree 14.

The pattern behind the list is uniform: a non-closed transitive component revives if and only if it can be pinned either through a small quotient (the cyclic family; $C_3 = S_4/V_4$ for $A_4$) or through an outer-twisted gluing with a genuine incidence geometry (the Fano plane for $\mathrm{PSL}(2,7)$). The same criteria predict the next revivals beyond the present catalogue boundary: $F_{39} = C_{13} \rtimes C_3$, the index-2 phantom inside the Paley group $F_{78}$ of degree 13, should first revive at $N = 16$ ($13 + 3$, its $C_6$ quotient pinned down to $C_3$ by a chiral triangle), and $\mathrm{PSL}(3,3)$ --- a non-closed transitive group of degree 13 with closure $S_{13}$ --- at $N = 26$, as the point--line incidence system of $\mathrm{PG}(2,3)$.

\section{Symmetry-Exact Visualization:\texorpdfstring{\\}{ }Diagonalizing the Topology Matrix}\label{sec:embedding}

\subsection{The Class-Weighted Topology Matrix}
Besides the $2^N \times 2^N$ Hamiltonian of Section~7, a spin system carries a second, far smaller symmetric matrix: the $N \times N$ \emph{topology matrix} $W$, obtained by assigning a generic weight $w_c$ --- pairwise distinct, fixed once --- to each coupling class $c$ and setting $W_{ij} = w_{c(ij)}$ (for the catalogue entries all vertices share one shift class, so the diagonal may be taken zero; distinct shift classes would enter as generic diagonal weights). Two properties make $W$ the bridge between the combinatorial catalogue and geometry. First, every $\sigma \in \mathrm{Aut}(G)$ preserves the classes, so its permutation matrix $P_\sigma$ commutes with $W$; conversely, because the weights are pairwise distinct, any vertex permutation commuting with $W$ preserves every class. Hence $\mathrm{Aut}(W) = \mathrm{Aut}(G)$ exactly: the spectrum of $W$ sees the true symmetry of the entry and nothing more. Second, by the spectral theorem the eigenspaces of $W$ are consequently $\mathrm{Aut}$-invariant, and the $N$-dimensional permutation representation of $\mathrm{Aut}$ decomposes orthogonally across them. Diagonalizing the topology matrix is the vertex-space shadow of the same commutation principle that block-diagonalizes the Hamiltonian in Section~7 --- one group, two matrices, two scales ($N$ versus $2^N$).

The geometric meaning of this diagonalization is fixed by the following observation, which makes the intuition ``a spin system is a distorted regular simplex'' literal.

\begin{proposition}[Euclidean realization]\label{prop:euclid}
The realizable symmetry groups of $N$-spin systems are exactly the isometry groups of full-dimensional $N$-point configurations in $\mathbb{R}^{N-1}$, viewed as permutation groups of the points.
\end{proposition}

\begin{proof}
The squared-distance matrix of the regular simplex is an interior point of the cone of Euclidean distance matrices: its centred Gramian is proportional to $I - J/N$, positive definite on $\mathbf{1}^{\perp}$. Every sufficiently small perturbation of the squared distances therefore remains the distance matrix of a configuration of full affine rank $N-1$. Perturbing along the edge classes of a $2^{*}$-closed group $H$ --- one small generic increment per class --- yields a configuration whose distance-preserving point permutations are precisely the automorphisms of the colouring, i.e.\ $H$. Conversely, for any full-dimensional configuration a distance-preserving permutation of the points extends uniquely to an isometry of $\mathbb{R}^{N-1}$, and the isometry group of the configuration is the automorphism group of its distance colouring --- a realizable group by the criterion of Section 2. (Physical couplings need not satisfy the metric axioms, but only the combinatorics of the equality classes matters, and the Euclidean model reproduces it in full.)
\end{proof}

Chirality is included: a twisted antiprism is a perfectly legal point configuration whose isometry group is the purely rotational $C_3$ --- the Euclidean incarnation of entry \texttt{6/o3.1}.

\subsection{Spectral Embedding in Three Dimensions}
The accompanying tool \texttt{topology2xyz} turns this into pictures. Let $U$ be an orthonormal eigenbasis of $W$ and select a union $S$ of \emph{whole} eigenspaces with total dimension 3; the rows of $U_S$ are the vertex coordinates. For every $\sigma \in \mathrm{Aut}$ the invariance of $S$ gives $P_\sigma U_S = U_S Q_\sigma$ with $Q_\sigma$ orthogonal: every graph automorphism acts on the drawing as an exact rigid motion, and since the Gram matrix $U_S^{\vphantom{T}} U_S^{T}$ is the invariant projector onto $S$, the configuration is canonical up to a global rotation. Among all eigenspace selections of total dimension 3 the tool picks the one maximizing the minimal inter-vertex distance, then exports the result as a pseudo-molecule (\texttt{.xyz}/\texttt{.mol2}: element = vertex orbit, bonds = the sparse coupling classes), directly viewable in standard chemistry viewers. Exactness is verified a posteriori by an orthogonal-Procrustes check --- the nearest orthogonal $Q_\sigma$ is recovered by an SVD and the residual $\|P_\sigma X - X Q_\sigma\|$ is at machine precision. Panels (a) and (b) of Figure~\ref{fig:embeddings} --- the twisted triangles with $\mathrm{Aut} = C_3$ and the cube with $|\mathrm{Aut}| = 48$ --- are such embeddings.

Two caveats are intrinsic rather than numerical, and the tool reports both. First, the 3D action is always isometric but need not be \emph{faithful}: distinct automorphisms may land on the same isometry --- e.g.\ the group $D_2 \times D_2$ acting on two tetrads, of order 16, exceeds the largest elementary abelian 2-subgroup of $\mathrm{O}(3)$ and cannot act faithfully in three dimensions; the tool counts the distinct isometries realized. Second, some automorphism groups are not three-dimensional point groups \emph{on the given vertex set} at all: no 3-dimensional invariant subspace separates the vertices. The tool then finds the smallest separating symmetric embedding in dimension $D > 3$ and folds the surplus coordinates into the third axis with small deterministic coefficients --- all vertices separate, and the drawing remains exact for the subgroup fixing the folded directions. The $C_9$ realization at $N = 12$ (orbits $9+3$, Figure~\ref{fig:embeddings}d) is of this kind: the generator advances the 9-ring by one step while turning the triangle by $120^{\circ}$, an action no rigid rotation of $\mathbb{R}^{3}$ can perform --- which is exactly why panel (d) is drawn schematically. In this sense the catalogue is an atlas of shapes, not merely a list of groups: every entry exports as a molecule-like object whose rigid symmetries are exactly its realizable group, and the failures of three-dimensional faithfulness are themselves informative, flagging the entries whose symmetry is intrinsically higher-dimensional.

\section{Hierarchical Methodology for Exact Block Diagonalization}

\subsection{Layer 0 (L0): Conservation of Total Spin Projection}
Before applying any graph-theoretical symmetry, the first stage of factorization relies on the fundamental physics of the isotropic spin Hamiltonian. The Hamiltonian $H$ strictly commutes with the total spin projection operator along the quantization axis, $F_z = \sum_{i=1}^N I_{z,i}$. This universal property ensures that the $2^N$-dimensional Hilbert space $\mathcal{H}$ naturally partitions into $N+1$ orthogonal subspaces (blocks), each characterized by a distinct magnetic quantum number $m \in \{-N/2, \dots, N/2\}$. While this L0 reduction is standard in all NMR simulation routines, it serves as the necessary foundation for the subsequent group-theoretical projections.

\subsection{Layer 1 (L1): Magnetic Equivalence and Schur--Weyl Duality}
The first non-trivial phase of symmetry reduction addresses the twin vertices of the weighted graph --- the strictly magnetically equivalent nuclei of Section 3.3. Rather than explicitly constructing symmetrized linear combinations in the full Hilbert space, we leverage Schur--Weyl duality \cite{weyl1939}. Each magnetic composite of $n_g$ spins carries a local symmetric group $S_{n_g}$, and its state space factorizes into irreducible representations of $S_{n_g}$, which map directly to the total-spin values of the composite particle \cite{banwell1963}. The Schur--Weyl viewpoint on nuclear-spin statistics has seen a systematic revival in molecular spectroscopy: restricting the $S_N$-decomposition to the permutation subgroup realized by rovibrational motion yields the nuclear-spin statistical weights of rigid molecules, combinatorially via the major-index branching rule for $S_N \downarrow C_m$ \cite{kubischta2026}. In the liquid-state setting of the present work the relevant subgroup is instead the full automorphism group of the coupling graph --- which, by Proposition~\ref{prop:ring}, is never purely cyclic on a single ring.

By analytically contracting these magnetically equivalent nuclei into effective composite particles, the Hamiltonian is projected onto a composite basis. This operation achieves a substantial reduction in matrix dimensionality before any global topological analysis is required. Integrating this L1 reduction directly into the computational core of quantum-mechanical simulators circumvents the generation of redundant states, dramatically accelerating the evaluation of transition matrix elements.

\subsection{Layer 2 (L2): Configuration Orbits and Isotypic Projections of the Factor Group}
Following the L1 contraction, the original complete graph $G$ is reduced to a smaller factor graph, whose residual symmetry is the factor group $\Gamma_q$ of Eq.~\eqref{eq:factor}. The factor-group reduction proceeds in two stages of increasing delicacy.

\emph{L2a: configuration orbits and orbit weight.} A configuration $c = (S_1, \dots, S_k)$ fixes the total spin of every composite particle; since each $S_g$ is an exact quantum number, $H$ is block-diagonal in $c$. The factor group permutes the particles and hence the configurations, and for $\sigma \in \Gamma_q$ the blocks $H_c$ and $H_{\sigma \cdot c}$ are related by a unitary permutation, so their spectra and $F^{+}$ matrix elements coincide. It therefore suffices to diagonalize a \emph{single representative per $\Gamma_q$-orbit} and weight its lineshape contribution by the orbit size (times the Schur--Weyl multiplicity). This reduction is exact for any factor group --- abelian or not --- and requires no symmetry-adapted basis at all.

\emph{L2b: residual symmetry of fixed configurations.} Configurations with a non-trivial stabilizer $\mathrm{Stab}(c) \le \Gamma_q$ retain genuine residual symmetry acting \emph{inside} the block, where the flip-flop terms couple different $m$-tuples at fixed total projection; here a symmetry-adapted reduction is unavoidable. We apply an isotypic projection of the L1-basis states over the irreducible representations of $\mathrm{Stab}(c)$, with a unified treatment of abelian and non-abelian stabilizers, including those possessing complex characters. Wigner projection operators are constructed systematically from the character tables; the projection acts at the level of isotypic components and does not fold the multiplicity of multidimensional irreducible representations, so no degenerate-partner alignment ambiguity arises. This guarantees an exact segregation of the Hamiltonian into independent blocks, mathematically ensuring the absence of off-diagonal elements between states belonging to different irreducible representations. For stabilizers whose characters are all real (integer $\pm 1$; e.g.\ $C_2$, $D_2$ and their products) the same basis is obtained more cheaply by a real simultaneous diagonalization of the commuting permutation operators, without any character table.

The character tables themselves are computed numerically from the group alone. The class-multiplication matrices span the centre of the group algebra, commute, and are simultaneously diagonalizable; their common eigenvectors are in bijection with the irreducible representations and carry the central characters, from which the full table is recovered. The joint eigenbasis is found by \emph{sequential} eigenspace splitting, extracting each eigenspace as an SVD null space --- a deterministic procedure, uniformly stable for complex-character cyclic groups, for non-abelian groups, and for the case that defeats the alternative single-random-combination approach: large abelian factor groups, where the number of classes equals the group order.

The taxonomy of Section~\ref{sec:taxonomy} records, for every catalogue entry, whether the automorphism group is \emph{ambivalent} (every element conjugate to its inverse, equivalently all characters real). For non-ambivalent groups the complex-conjugate pairs of one-dimensional representations produce pairs of blocks with identical spectra --- for a real symmetric Hamiltonian only one block of each pair needs to be diagonalized.

\subsection{Worked Example}
Table~\ref{tab:blocks} illustrates the full L0--L2 reduction for the catalogue entry \texttt{8/o6.5} --- the chiral 8-spin system $C_3\langle A_3^\prime M_3^\prime \rangle \ast X_2$ (two twisted triangles plus one magnetic pair; $|\mathrm{Aut}| = 6$, twins $2$, $|\Gamma_q| = 3$). L1 contracts the pair into a composite with total spin $j \in \{1, 0\}$; since the rotated particles are singletons, every configuration is $\Gamma_q$-fixed, and the reduction is the L2b projection onto the three characters $\chi_0, \chi_1, \chi_2$ of $\mathrm{Stab} = \Gamma_q = C_3$, where $\chi_1, \chi_2$ are complex conjugates. The largest diagonalization block shrinks from $\binom{8}{4} = 70$ (L0 alone) to $18$; moreover, the $\chi_1$ and $\chi_2$ blocks are spectrally identical, so only one of each pair is computed.

\begin{table}[t]
\centering
\caption{Block dimensions for entry \texttt{8/o6.5} ($2^8 = 256$ states) after the L0 ($F_z$), L1 (composite pair, total spin $j$) and L2 ($C_3$ characters $\chi_k$) reductions. The $\chi_1$ and $\chi_2$ blocks are complex-conjugate and have identical spectra.}
\label{tab:blocks}
\begin{tabular}{c|ccc|ccc}
\toprule
 & \multicolumn{3}{c|}{$j = 1$} & \multicolumn{3}{c}{$j = 0$} \\
$m$ & $\chi_0$ & $\chi_1$ & $\chi_2$ & $\chi_0$ & $\chi_1$ & $\chi_2$ \\
\midrule
$\pm 4$ & 1 & 0 & 0 & 0 & 0 & 0 \\
$\pm 3$ & 3 & 2 & 2 & 1 & 0 & 0 \\
$\pm 2$ & 8 & 7 & 7 & 2 & 2 & 2 \\
$\pm 1$ & 15 & 13 & 13 & 5 & 5 & 5 \\
$0$ & 18 & 16 & 16 & 8 & 6 & 6 \\
\bottomrule
\end{tabular}
\end{table}

The resulting hierarchical reduction (L0 $\rightarrow$ L1 $\rightarrow$ L2) is mathematically exact, invokes no physical approximations, and enables the computation of the total NMR lineshape for complex, highly symmetric multi-spin topologies with maximal efficiency.

\subsection{Exactness and Validation}
The pipeline factorizes into independently checkable stages --- basis symmetrization, block representation, diagonalization, and intensity assembly --- separated by two invariants. First, the multiset of all Hamiltonian eigenvalues, taken with the orbit, Schur--Weyl and irrep multiplicities, is a similarity invariant and must reproduce the spectrum of the dense $2^N$ reference engine exactly; this validates every stage up to diagonalization, independently of eigenvector ambiguities in degenerate subspaces. Second, the total lineshape must match the reference; the spectrum alone is necessary but not sufficient, since a defect in the $F^{+}$ matrix elements can hide behind perfectly correct eigenvalues, and the pair of invariants localizes any discrepancy to a single stage.

The catalogue of Section 4 doubles as the exhaustive test set. On all 90 realizable 8-spin topologies the reduced engine reproduces both the spectrum (to ${\sim}10^{-13}$) and the lineshape of the dense reference for every entry --- including the complex-character ($C_3$, $C_4$, $C_6$) and non-abelian groups that no single real simultaneous diagonalization can handle, which is precisely why case completeness matters. The group-theoretic layer scales further: on all 2866 realizable 14-spin classes the engine reproduces the order identity of Eq.~\eqref{eq:factor} for every class, with the pruned automorphism search matching the reference particle-permutation group exactly on every reference-tractable class while remaining polynomial where naive generate-and-test would require up to $14! \approx 8.7 \cdot 10^{10}$ permutations.

\subsection{Computational Advantages of Hierarchical Reduction}
The three-layer reduction methodology (L0--L2) provides a complete algorithmic solution for the exact block diagonalization of the isotropic spin Hamiltonian. Unlike approximate simulation methods, this approach introduces no physical simplifications, preserving the exact eigenvalues and transition matrix elements. The contraction of magnetically equivalent nuclei using Schur--Weyl duality (L1), combined with the orbit-weight deduplication of configurations and the isotypic projection over the irreducible representations of the factor group (L2a/L2b), partitions the Hamiltonian into the minimal symmetry-dictated blocks (Table~\ref{tab:blocks}), eliminates the redundancy associated with off-diagonal elements between states of different symmetry, and resolves the difficulties traditionally associated with non-abelian groups and complex characters; for non-ambivalent groups, conjugate representation pairs halve the remaining work.

\section{Conclusion and Outlook}
In this work we have given an exact realizability criterion for spin-system symmetries: the symmetry groups of scalar-coupled spin systems are precisely the symmetrized-2-closed permutation groups. On a single ring, chirality is impossible --- rotational symmetry always drags in a reflection; yet chiral spin systems exist, their minimal sizes are governed by the quotient-ring mechanism of Eq.~\eqref{eq:mu}, and the catalogue realizations $C_8$ and $C_9$ at $N = 12$ show that the true minima can undercut naive expectations. The complete taxonomy up to $N = 14$ --- $a(14) = 2866$ essentially distinct realizable classes, $\sum_{N=3}^{14} a(N) = 6112$ catalogue entries in all --- with its canonical identifiers, structural formulas and machine fingerprints, provides an exhaustive reference database for both theoretical and experimental magnetic resonance. Reaching $N = 14$ already required the orbit-partition decomposition of Section~\ref{sec:decomposition}, the direct enumeration having outgrown the subgroup lattice of $S_N$; the realizability criterion itself is size-independent, so the boundary of the catalogue is now computational rather than conceptual.

The primary avenue for future development is the direct integration of this hierarchical reduction architecture --- the automated Wigner projection and composite particle mapping routines --- into high-performance total lineshape analysis software such as ANATOLIA \cite{cheshkov2018} and its X-factorization extension \cite{nichugovskiy2026}. Merging this rigorous graph-theoretical framework with optimized parallel code will enable the routine analysis of total lineshapes for unprecedentedly large and complex spin systems, opening new frontiers in structural chemical analysis.

\section*{Data and Code Availability}
{\sloppy
The complete catalogues of realizable topologies for $N = 3$--$14$, the machine taxonomy (\texttt{catalog\_taxonomy.tsv}, 6112 entries), the cross-referenced and renamed classical 8-spin catalogue, the enumeration codes (single-node and distributed MPI versions, together with the GAP orbit-partition decomposition used for $N = 14$), the taxonomy classifier and the symmetric 3D-embedding tool (\texttt{topology2xyz}, Section~\ref{sec:embedding}) are available in the SpinSystemsSymmetry folder of the ANATOLIA repository on GitHub \cite{anatolia}.\par}

\end{document}